\documentclass[11pt]{article}
\pdfoutput=1

\usepackage[top=2.8cm , bottom=1.5cm , right=3cm , left=2.5cm]{geometry}

\usepackage{epsfig}
\usepackage{graphicx}
\usepackage[latin1]{inputenc}
\usepackage[T1]{fontenc}
\usepackage{mathptmx}
\usepackage{url}

\usepackage{bibspacing}
\usepackage{amsfonts,latexsym,amssymb}

\usepackage[ruled,vlined]{algorithm2e}
\usepackage{algpseudocode}
\usepackage{lhelp}
\usepackage{verbatim}
\usepackage{fancyvrb}
\usepackage{lipsum}
\usepackage{pstricks,pst-node,pst-grad,pst-plot}

\def\Q{\mbox{${\mathbb Q}$}}

\newcommand{\maple}{{\sc Maple}}
\newcommand{\Maple}{{\sc Maple}}

\def\delta{\mbox{{$\Delta$}}}

\def\RegularChains{\mbox{{\tt RegularChains}}}

\newcommand{\nvcc}{\mbox{\sc nvcc}}
\newcommand{\openmp}{\mbox{\sc OpenMP}}

\newcommand{\meta}{\mbox{\sc MetaFork}}
\newcommand{\ppcg}{\mbox{\sc PPCG}}
\newcommand{\ctocuda}{\mbox{\sc C-to-CUDA}}
\newcommand{\hicuda}{\mbox{\sc HiCUDA}}
\newcommand{\cudachill}{\mbox{\sc CUDA-CHiLL}}

\newcommand{\llvm}{\mbox{\sc LLVM}}
\newcommand{\ptx}{\mbox{\sc PTX}}

\newcommand{\openacc}{\mbox{\sc OpenACC}}
\newcommand{\opencl}{\mbox{\sc OpenCL}}

\newcommand{\cuda}{\mbox{\sc CUDA}}

\newcommand{\hidetext}[1]{\mbox{ \ }}

\newtheorem{lemma}{Lemma}

\newtheorem{definition}{Definition}

\newtheorem{remark}{Remark}

\newenvironment{proof}{ {\noindent} {\sc Proof} $\rhd$}{$\lhd$}

\newif\ifarticle
\articletrue

\newif\ifextendedabstract
\extendedabstractfalse

\begin{document}

\title{Comprehensive Optimization of Parametric Kernels for Graphics Processing Units}

\author{\begin{tabular}{cccc}
Xiaohui Chen & Marc Moreno Maza & Jeeva Paudel & Ning Xie \\
AMD          & U. Western Ontario     & IBM Canada Ltd & Huawei Technologies Canada \\
Markham, ON  & London, ON       & Markham, ON    & Markham, ON \\
Canada       & Canada           & Canada         & Canada \\
\end{tabular}
}
\maketitle

\begin{abstract}
This work deals with the optimization of computer programs targeting 
Graphics Processing Units (GPUs). The goal is to lift, from 
programmers to optimizing compilers, the heavy burden of determining 
program details that are dependent on the hardware characteristics. 
The expected benefit is to improve robustness, portability and 
efficiency of the generated computer programs. We address these 
requirements by: 
\begin{itemizeshort}
\item treating machine and program parameters as unknown symbols during 
code generation, and 
\item generating optimized programs in the form of a case discussion, 
based on the possible values of the machine and program parameters. 
\end{itemizeshort}
By taking advantage of recent advances in the area of computer algebra, 
preliminary experimentation yield promising results.
\end{abstract}

\section{Introduction}

It is well-known that the advent of hardware acceleration technologies
(multicore processors, graphics
processing units, field programmable gate arrays) provide vast
opportunities for innovation in computing.  In particular, GPUs
combined with {\em low-level heterogeneous programming models}, such
as {\cuda} (the {\em Compute Unified Device Architecture}, 
see~\cite{Nickolls:2008:SPP:1365490.1365500,cuda2015}), 
brought super-computing to the level of the desktop computer.
However, these low-level programming models carry notable challenges, even to
expert programmers. Indeed, fully exploiting the power of hardware
accelerators by writing {\cuda} code often requires significant code
optimization effort. While such effort can yield high
performance, it is desirable for many programmers to avoid the
explicit management of the hardware accelerator, e.g. data transfer
between host and device, or between memory levels of the device.
To this end, {\em high-level} models for accelerator programming, 
notably {\openmp}~\cite{dagum1998openmp,openmp2015}
and {\openacc}~\cite{DBLP:conf/lcpc/TianX0YCC13,openacc2015}, have become an important research
direction. With these models, programmers only need to annotate their
C/C++ (or FORTRAN) code to indicate which portion of code is to be
executed on the device, and how data is mapped between host and device.

In {\openmp} and {\openacc}, the division of the work between thread
blocks within a grid, or between threads within a thread block, can be
expressed in a loose manner, or even ignored.  This implies that code
optimization techniques must be applied in order to derive
efficient {\cuda} code. Moreover, existing software packages (e.g. 
{\ppcg}~\cite{DBLP:journals/taco/VerdoolaegeJCGTC13},
{\ctocuda}~\cite{Baskaran:2010:ACC:2175462.2175482},
{\hicuda}~\cite{Han:2009:HCH:1513895.1513902},
{\cudachill}~\cite{Khan:2013:SAC:2400682.2400690}) for generating
{\cuda} code from annotated C/C++ programs, either let the user choose,
or make assumptions on, the characteristics of the targeted hardware, 
and on how the work
is divided among the processors of that device. These choices
and assumptions
limit {\em code portability} as well as opportunities for 
{\em code optimization}.  This latter fact will be illustrated
with dense matrix multiplication, through Figures~\ref{fig:metamul} and~\ref{fig:cudamul},
as well as Table~\ref{tab:mm}.

To deal with these challenges in translating annotated C/C++ programs
to {\cuda}, we propose in~\cite{DBLP:conf/cascon/ChenCKMX15} to
generate {\em parametric {\cuda} kernels}, that is, {\cuda} kernels
for which program parameters (e.g. number of threads per thread block) and
machine parameters (e.g. shared memory size) are symbolic entities instead of
numerical values.
Hence, the values of these parameters need not to be known 
during code generation: machine parameters can be looked up 
when the generated code is loaded on the target machine,
while program parameters can be deduced, for instance, by auto-tuning.
See Figure~\ref{fig:cudamul} for an example of parametric {\cuda} kernels.
A proof-of-concept implementation, presented
in~\cite{DBLP:conf/cascon/ChenCKMX15} and publicly
available\footnote{\url{www.metafork.org}}, uses another high-level
model for accelerator programming, called {\meta}, that we introduced
in~\cite{DBLP:conf/iwomp/ChenMSU14}.  The experimentation shows that
the generation of parametric {\cuda} kernels can lead to significant
performance improvement w.r.t. approaches based on the generation of
{\cuda} kernels that are {\em not} parametric.  Moreover, for certain
test-cases, our experimental results show that the optimal choice for
program parameters may depend on the input data size. For instance,
the timings gathered in Table~\ref{tab:mm} show that 
the format of the 2D thread-block
of the best {\cuda}
kernel that we could generate is $16 \times 8$
for matrices of order $2^{10}$ and $32 \times 8$ for
matrices of order $2^{11}$.  Clearly, parametric {\cuda} kernels are
well-suited for this type of test-cases.

In this paper, our goal is to enhance the 
framework initiated in~\cite{DBLP:conf/cascon/ChenCKMX15} by
generating {\em optimized} parametric {\cuda} kernels.
As we shall see, this can be done 
in the form of a case discussion, based on the possible values
of the machine and program parameters.
The output of a procedure generating optimized parametric 
{\cuda} kernels will be called a {\em comprehensive 
parametric {\cuda} kernel}. 
A simple example is shown on Figure~\ref{fig:compparkern}.
In broad terms, this is  a decision tree where:
\begin{enumerateshort}
\item each internal node 
      is a Boolean condition on the machine  and program parameters, and
\item each leaf is a {\cuda} program ${\cal P}$, optimized w.r.t. 
      prescribed criteria and optimization techniques, under 
      the conjunction of the conditions along the path
      from the root of the tree to ${\cal P}$.
\end{enumerateshort}
The intention, with this concept, is to automatically generate
optimized {\cuda} kernels from annotated C/C++ code without knowing
the numerical values of some or even any of
the machine and program parameters.
This naturally leads to case distinction depending on
the values of those parameters, which materializes into a disjunction
of conjunctive non-linear polynomial constraints.  Symbolic
computation, aka computer algebra, is the natural framework for
manipulating such systems of constraints; our {\RegularChains}
library\footnote{This library, shipped
with the commercialized computer algebra system {\sc Maple}, 
is freely available at \url{www.regularchains.org}.} 
provides the appropriate algorithmic tools for that task.

Other  research groups  have  approached the  questions  of {\em  code
portability} and  {\em code  optimization} in  the context  of {\cuda}
code  generation   from  high-level   programming  models.  They  use
techniques like auto-tuning~\cite{grauer2012auto,Khan:2013:SAC:2400682.2400690},
dynamic instrumentation~\cite{kistler2003continuous} or both~\cite{song2015automated}.
Rephrasing~\cite{Khan:2013:SAC:2400682.2400690}, ``those techniques
explore empirically different data  placement and thread/block mapping
strategies, along  with other  code generation  decisions, thus 
facilitating the finding of a high-performance solution.''

In the case of auto-tuning techniques, which have been used
successfully in the celebrated projects 
ATLAS~\cite{DBLP:conf/ppsc/WhaleyD99},
FFTW~\cite{DBLP:conf/icassp/FrigoJ98}, and
SPIRAL~\cite{DBLP:journals/ijhpca/PuschelMSXJPVJ04}, part of the code
optimization process is done {\em off-line}, that is, the input code
is analyzed and an optimization strategy (i.e a sequence of composable
code transformations) is generated, and then applied on-line (i.e. on
the targeted hardware).
We propose to push this idea further by applying the optimization strategy
off-line, thus, 
even before the code is loaded on the targeted hardware.

\ifextendedabstract
\else
Let us illustrate, with an example, 
\fi
the notion of comprehensive
parametric {\cuda} kernels, along with a procedure to generate them.
Our input is the for-loop nest of Figure~\ref{fig:tiledmetaforloop}
 which
computes the sum of two matrices {\tt b} and {\tt c} of order $N$
using a blocking strategy; each matrix is divided into
blocks of format ${\tt B0} \times {\tt B1}$.
This input code is 
annotated for parallel execution in the {\meta} language.
The body of the statement \textcolor{red}{\tt meta\_schedule} is meant
to be offloaded to a GPU device and each
\textcolor{red}{meta\_for} loop is a parallel for-loop
where all iterations can be executed concurrently.

\begin{figure}[!htb]
\begin{Verbatim}[commandchars=\\\{\}, fontsize=\footnotesize]
  int dim0 = N/B0, dim1 = N/(2*B1);
  \textcolor{red}{meta_schedule} \{
    \textcolor{red}{meta_for} (int v = 0; v < dim0; v++)
      \textcolor{red}{meta_for} (int p = 0; p < dim1; p++)
        \textcolor{red}{meta_for} (int u = 0; u < B0; u++) 
          \textcolor{red}{meta_for} (int q = 0; q < B1; q++) \{
            int i = v * B0 + u;
            int j = p * B1 + q;
            if (i < N && j < N/2) \{
              c[i][j] = a[i][j] + b[i][j]; 
              c[i][j+N/2] = 
                 a[i][j+N/2] + b[i][j+N/2]; 
            \}
          \}
  \}
\end{Verbatim}
\caption{A \textcolor{red}{\tt meta\_for} loop nest for adding two matrices.}
\label{fig:tiledmetaforloop}
\end{figure}

We make the following simplistic assumptions for the 
translation of this for-loop nest to {\cuda}.
\begin{enumerateshort}
\item The target machine has two parameters: the maximum number $R$ 
      of registers per thread, and the maximum number $T$ of threads per 
      thread-block; all other hardware limits are ignored.
\item The generated kernels depend on two program parameters, $B_0$ and
      $B_1$, which define the format of a 2D thread-block.
\item The optimization strategy (w.r.t. register usage per thread) 
      consists in reducing the work per thread
      (by reducing loop granularity).
\end{enumerateshort}
A possible comprehensive parametric {\cuda} kernel 
is given by the pairs $(C_1, K_1)$ and $(C_2, K_2)$, where $C_1, C_2$ are
two sets of algebraic constraints on the parameters and $K_1, K_2$ are
two {\cuda} kernels that are optimized under the 
constraints respectively given by $C_1, C_2$, 
see Figure~\ref{fig:compparkern}.
The following computational steps yield the 
pairs $(C_1, K_1)$ and $(C_2, K_2)$.
\begin{enumerateshort}
\item[(S1)]  The {\meta} code is mapped to an
      intermediate representation (IR)
      say that of {\llvm}\footnote{
       Quoting Wikipedia: 
       ``The {\llvm} compiler infrastructure project 
       (formerly Low Level 
       Virtual Machine~\cite{Lattner:2004:LCF:977395.977673,Bertolli:2014:CGT:2688361.2688364})
       is a framework for developing 
       compiler front ends and back ends''.}, or alternatively, to
       {\ptx}\footnote{The {\em Parallel Thread Execution} 
       (PTX)~\cite{ptx2015}
        is the pseudo-assembly language to which
       {\cuda} programs are compiled by NVIDIA's {\nvcc} compiler.
       {\ptx} code can also be generated from (enhanced) {\llvm} IR,
       using {nvptx} back-end~\cite{nvptx},
       following the work of~\cite{Rhodin2010}.} code.
\item[(S2)] Using this IR (or {\ptx}) code, one {\em estimates} the 
      number of registers that a thread requires;
     on this example, using {\llvm} IR, we obtain an estimate of $14$.
\item[(S3)] Next, we apply the optimization strategy, yielding 
      a new IR (or {\ptx}) code, for which register pressure reduces to $10$.
      Since no other optimization techniques are considered,
      the procedure stops with the result shown on 
      Figure~\ref{fig:compparkern}.
\end{enumerateshort}
Note that, strictly speaking, the kernels $K_1$ and $K_2$ on
Figure~\ref{fig:compparkern} should be given by {\ptx} code.  But for
simplicity, we are presenting them by counterpart {\cuda} code.

\begin{figure}[!htb]
\begin{small}
$C_1 : \left\{ \begin{array}{l} B_0 \times B_1 \leq T \\ 14 \leq R \end{array} \right.   $ 
\end{small}

\begin{scriptsize}
\begin{verbatim}
__global__  void K1(int *a, int *b, int *c, int N, 
                                 int B0, int B1) {
     int i = blockIdx.y * B0 + threadIdx.y;
     int j = blockIdx.x * B1 + threadIdx.x;
     if (i < N && j < N/2) {
        a[i*N+j] = b[i*N+j] + c[i*N+j];
        a[i*N+j+N/2] = b[i*N+j+N/2] + c[i*N+j+N/2];
     }
}
dim3 dimBlock(B1, B0);
dim3 dimGrid(N/(2*B1), N/B0);
K1 <<<dimGrid, dimBlock>>> (a, b, c, N, B0, B1);
\end{verbatim}
\end{scriptsize}

\begin{small}
$C_2 : \left\{ \begin{array}{l} B_0 \times B_1 \leq T \\ 10 \leq R < 14 \end{array} \right.  $
\end{small}

\begin{scriptsize}
\begin{verbatim}
__global__  void K2(int *a, int *b, int *c, int N, 
                                 int B0, int B1) {
     int i = blockIdx.y * B0 + threadIdx.y;
     int j = blockIdx.x * B1 + threadIdx.x;
     if (i < N && j < N)
        a[i*N+j] = b[i*N+j] + c[i*N+j];
}
dim3 dimBlock(B1, B0);
dim3 dimGrid(N/B1, N/B0);
K2 <<<dimGrid, dimBlock>>> (a, b, c, N, B0, B1);
\end{verbatim}
\end{scriptsize}
\caption{A comprehensive parametric {\cuda} kernel for matrix addition.}
\label{fig:compparkern}
\end{figure}

\ifextendedabstract
\else
This paper is organized as follows.  In
Section~\ref{sec:parametric_kernels}, we review the notion of a {\em
parametric {\cuda} kernel} through an example.
In Section~\ref{sec:comprehensive_optimization},
we introduce the notion of {\em comprehensive optimization}
of a code fragment together with an algorithm for computing it.
In Section~\ref{comprehensive_translation}, we explain
how this latter notion applies to the generation
of parametric {\cuda} kernels generated from 
a program written in  a high-level accelerator model namely
{\meta}.
Finally, experimental results are provided in
Section~\ref{sec:experimentation}.
\fi

\section{Parametric kernels}
\label{sec:parametric_kernels}

We review and illustrate the notion of a {\em parametric {\cuda}
  kernel} (introduced in~\cite{DBLP:conf/cascon/ChenCKMX15}) with an
example: computing the product of two dense square matrices of order
{\tt n}.
Figure~\ref{fig:metamul} shows a code snippet,
expressed in the {\meta} language,  performing 
a {\em blocking strategy}.
Each iteration of the  {\em parallel for-loop nest} (i.e. the 4
{\tt meta\_for} nested loops) computes 
{\tt s} coefficients of the product matrix.
The blocks in the matrices {\tt a}, {\tt b}, {\tt c}
have format ${\tt B0} \times {\tt B0}$, 
${\tt B0} \times ({\tt ub1} \cdot {\tt s})$, 
${\tt B0} \times ({\tt ub1} \cdot {\tt s})$.
Note that memory accesses to {\tt a}, {\tt b}, {\tt c}
are coalesced in both codes.

\begin{figure}[htb]
\begin{Verbatim}[commandchars=\\\{\}, fontsize=\scriptsize]
assert(B0 <= ub1 * s);
int dim0 = n / B0, dim1 = n / (ub1 * s);
\textcolor{red}{meta_schedule} \{
  \textcolor{red}{meta_for} (int i = 0; i < dim0; i++) 
    \textcolor{red}{meta_for} (int j = 0; j < dim1; j++) 
      for (int k = 0; k < n / B0; k++) 
        \textcolor{red}{meta_for} (int v = 0; v < \textcolor{blue}{B0}; v++) 
          \textcolor{red}{meta_for} (int u = 0; u < \textcolor{blue}{ub1}; u++) \{
            int p = i * B0 + v;
            int q = j * ub1 * s + u;
            for (int z = 0; z < B0; z++) 
              for (int w = 0; w < \textcolor{blue}{s}; w++) \{
                int x = w * ub1;
                c[p][q+x] += 
                  a[p][B0*k+z] * b[B0*k+z][q+x];
              \}
          \}
\}
\end{Verbatim}
\caption{{\meta} matrix multiplication using a blocking strategy
in {\meta}.}
\label{fig:metamul}
\end{figure}

\begin{figure}[!htb]
\begin{Verbatim}[commandchars=\\\{\}, fontsize=\scriptsize]
__global__ void kernel0(int *a, int *b, int *c, int 
   n, int dim0, int dim1, int \textcolor{blue}{B0}, int \textcolor{blue}{ub1}, int \textcolor{blue}{s}) \{
  int b0 = blockIdx.y, b1 = blockIdx.x;
  int t0 = threadIdx.y, t1 = threadIdx.x;
  int private\_p, private\_q;
\color{blue}  assert(B_0 == B0); assert(B_1 == ub1 * s);
  \_\_shared\_\_ int shared\_a[{B\_0}][B\_0];
  \_\_shared\_\_ int shared\_b[B\_0][{B\_1}];
  int private\_c[1][{S}];  \color{blue}assert(S == s);

  for (int c0 = b0; c0 < dim0; c0 += 256)
    for (int c1 = b1; c1 < dim1; c1 += 256) \{
      private\_p = ((c0) * (B0)) + (t0);
      private_q = ((c1) * (ub1 * s)) + (t1);
      for (int c5 = 0; c5 < \textcolor{blue}{S}; c5 += 1)
        if (n >= private\_p + 1 \&\& 
        n >= private_q + (c5) * (ub1) + 1)
          private\_c[0][c5] = c[(private\_p) * n + 
                  (private_q + (c5) * (ub1))];
      for (int c2 = 0; c2 < n / B0; c2 += 1) \{
        if (t1 < B0 \&\& n >= private\_p + 1)
          shared\_a[t0][t1] = 
            a[(private\_p) * n + (t1 + B0 * c2)];
        for (int c5 = 0; c5 < \textcolor{blue}{S}; c5 += 1)
          if (t0 < B0 \&\& 
          n >= private\_q + (c5) * (ub1) + 1)
            shared\_b[t0][(c5) * (ub1) + t1] = 
               b[(t0 + B0 * c2) * n + 
               (private\_q + (c5) * (ub1))];
        \_\_syncthreads();
        for (int c6 = 0; c6 < B0; c6 += 1)
          for (int c5 = 0; c5 < \textcolor{blue}{S}; c5 += 1)
            private\_c[0][c5] += 
                (shared\_a[t0][c6] * 
                shared\_b[c6][c5 * ub1 + t1]);
        \_\_syncthreads();
      \}
      for (int c5 = 0; c5 < \textcolor{blue}{S}; c5 += 1)
        if (n >= private\_p + 1 \&\& 
        n >= private_q + (c5) * (ub1) + 1)
          c[(private\_p) * n + 
          (private_q + (c5) * (ub1))] =
              private\_c[0][c5];
      \_\_syncthreads();
    \}
\}
\end{Verbatim}
\caption{{\cuda} kernel generated from 
a {\tt meta\_schedule} statement in the {\meta} language.}
\label{fig:cudamul}
\end{figure}

Figure~\ref{fig:cudamul} shows a {\cuda} kernel code generated
from the {\meta} code snippet of Figure~\ref{fig:metamul}.
Observe that {\tt kernel0} takes the {\em program parameters}
{\tt B0} and {\tt ub1} as arguments, whereas
non-parametric {\cuda} kernels usually 
only take data parameters (here $a, b, c, n$) as input arguments. 
Note also that, in order to allocate memory for the shared arrays 
{\tt shared\_a}, {\tt shared\_b}, {\tt shared\_c},
we predefine the names 
{\tt B\_0}, {\tt B\_1} as macros and specify their values
at compile time. Note that the assert statements 
ensure that {\tt B0}, {\tt ub1} match {\tt B\_0}, {\tt B\_1}.

To conclude with this example, we gather in Table~\ref{tab:mm} 
speedup factors w.r.t. a highly optimized  serial C program
implementing the same blocking strategy.
The numbers in {\bf bold} fonts correspond to the best speedup factors
by a parametric kernel on a given input size.
We observe that:
\begin{enumerateshort}
\item for $s = 4$, ${\tt ub1} = 16$, ${\tt B0} = 8$
when $n = 2^{10}$, and
\item for $s = 4$, ${\tt ub1} = 32$, ${\tt B0} = 8$
when $n = 2^{11}$,
\end{enumerateshort}
the parametric kernel of Figure~\ref{fig:cudamul}
provides the best results.

\begin{table}[htb]
\begin{tabular}{|c|c|c|c|c|}
\hline
{Thread-block $\backslash$ Input} & \multicolumn{2}{|c|}{$2^{10} * 2^{10}$} & \multicolumn{2}{|c|}{$2^{11} * 2^{11}$} \\
\hline
(ub1, B0)  &s = 2	&s = 4 & s = 2 & s = 4 \\
\hline
(16, 4)	&95 &128  &90 &119	\\
(32, 4)	&128 &157 &125 &144 \\
(64, 4)	&111  &145  &105 &132	\\
(8, 8)	&131 &151 &126 &146	\\
(16, 8)	&164 &\textbf{194}  &159 &188 \\
(32, 8)	&163  &187  &158 &\textbf{202}	\\
(64, 8)	&94 &143  &104 &135	\\
\hline
B0 & \multicolumn{4}{|c|}{Register usage for s = 4} \\
\hline
4 & \multicolumn{4}{c|}{38} \\
8 &  \multicolumn{4}{c|}{34} \\
\hline
\end{tabular}
\caption{\rm Speedup factors on an NVIDIA Tesla M2050 
for our 
kernel generated by {\meta} with compilation flag {\tt --maxrregcount=40}.
}
\label{tab:mm}
\end{table}

\section{Comprehensive Optimization}
\label{sec:comprehensive_optimization}

We consider a code fragment written in one of the linguistic
extensions of the C language targeting a computer device, for
instance, a hardware accelerator.  We assume that some, or all, of the
hardware characteristics of this device are unknown at compile time.
However, we would like to optimize our input code fragment w.r.t
prescribed resource counters (e.g. memory usage) and performance
counters (e.g. occupancy on a GPU device).  To this end, the hardware
characteristics of this device, as well as the program and data
parameters of the code fragment, are treated as symbols. From there,
we generate polynomial constraints (with those symbols as
indeterminate variables) so as to $(i)$ensure that sufficient
resources are available to run the transformed code, and $(ii)$
attempt to improve the code performance.

\subsection{Hypotheses on the input code fragment}  %
\label{sec:fragment}
We consider a sequence  ${\cal S}$ of statements
from the C programming language and introduce the following.

\begin{definition}
\label{defi:parameterOfCodeFragment}
We call {\em parameter} of ${\cal S}$ any scalar variable
that is $(i)$ read in ${\cal S}$ at least once, and
$(ii)$ never written in ${\cal S}$.
We call {\em data} of ${\cal S}$ any non-scalar variable (e.g. array)
that is not initialized but possibly overwritten within ${\cal S}$.
If a parameter of  ${\cal S}$ gives a dimension size of a data
of ${\cal S}$, then this parameter is called a {\em data parameter};
otherwise, it is simply called a {\em program parameter}.
\end{definition}

We denote by $D_1, \ldots, D_u$ and $E_1, \ldots, E_v$ 
the data parameters and program parameters 
of ${\cal S}$, respectively.

We make the following assumptions on  ${\cal S}$:
\begin{enumerateshort}
\item[(H1)] All parameters are assumed to be non-negative integers.
\item[(H2)] We assume that ${\cal S}$ can be viewed as the body of
            a valid C function having the parameters and data of  ${\cal S}$
            as unique arguments.
\end{enumerateshort}

The sequence of statements $S$ can be the body of a 
kernel function in {\cuda}. 
In the kernel code of Figure~\ref{fig:cudamul}, {\tt B0} and {\tt ub1}
are program parameters while {\tt a}, {\tt b} and {\tt c} are the data,
and that {\tt n} is the data parameter.

\subsection{Hardware resource limits and performance measures}
\label{sec:limits}

We denote by $R_1, \ldots, R_s$ the {\em hardware resource limits} of
the targeted hardware device. Examples of these quantities for 
the NVIDIA Kepler micro-architecture are
the maximum number of registers to be allocated per 
       thread, and
the maximum number of shared memory words to be allocated  
      per thread-block.
We denote by ${P}_1, \ldots, {P}_t$ the {\em performance measures}
of a program running on the device.
These are dimensionless quantities defined as percentages.
Examples of these quantities for the NVIDIA Kepler
micro-architecture are
the SM occupancy (that is, the ratio of the number of active warps to the
       maximum number of active warps)
and he cache hit rate in an streaming multi-processor (SM).

For a given hardware device, $R_1, \ldots, R_s$ are 
positive integers, and  each of them is
the maximum value of a hardware resource.
Meanwhile, ${P}_1, \ldots, {P}_t$ are rational numbers 
between $0$ and $1$.
However, for the purpose of writing code portable across a 
variety of devices with similar characteristics, the quantities
$R_1, \ldots, R_s$ and ${P}_1, \ldots, {P}_t$ will be treated as unknown  
and independent variables.
These hardware resource limits and performance measures
will be called the {\em machine parameters}.

Each function $K$ (and, in particular, our input code fragment ${\cal
  S}$) written in the C language for the targeted hardware device has
{\em resource counters} $r_1, \ldots, r_s$ and {\em performance
  counters} $p_1, \ldots, p_t$ corresponding, respectively, to $R_1,
\ldots, R_s$ and ${P}_1, \ldots, {P}_t$.
In other words, the quantities $r_1, \ldots, r_s$ are the amounts of
resources, corresponding to $R_1, \ldots, R_s$, respectively, that $K$
requires for executing.
Similarly, the quantities $p_1, \ldots, p_t$ are the
performance measures, corresponding to ${P}_1, \ldots, {P}_t$,
respectively, that $K$ exhibits when executing.
Therefore, the inequalities $0 \leq r_1 \leq R_1$, \ldots, $0 \leq r_s
\leq R_s$ must hold for the function $K$ to execute correctly.
Similarly, $0 \leq p_1 \leq 1$, \ldots, $0 \leq p_t \leq 1$ are
satisfied by the definition of the performance measures.

\begin{remark}
{\rm 
We note that $r_1, \ldots, r_s$, $p_1, \ldots, p_t$ may be numerical
values, which we can assume to be non-negative rational numbers.  This
will be the case, for instance, for the minimum number of registers
required per thread in a thread-block.
The resource counters $r_1, \ldots, r_s$
may also be polynomial expressions whose indeterminate variables 
can be program parameters
(like the dimension sizes of a thread-block or grid)
or data parameters (like the input data sizes).
Meanwhile, the performance counters $p_1, \ldots, p_t$ may further depend
on the hardware resource limits (like the maximum number of 
active warps supported by an SM).
To summarize,  we observe that  $r_1, \ldots, r_s$ are  polynomials in
${\Q}[D_1, \ldots, D_u, E_1, \ldots,  E_v]$ and $p_1, \ldots, p_t$ are
rational functions where numerators and denominators are in ${\Q}[D_1,
  \ldots, D_u, E_1, \ldots, E_v, R_1, \ldots, R_s]$.  
Moreover, we can assume that the denominators of those functions
are positive.
}
\end{remark}

On Figure~\ref{fig:compparkern}, $R$ and $T$ are
machine parameters while $B_0$ and $B_1$ are program parameters.  The
constraints displayed on Figure~\ref{fig:compparkern} are polynomials
in $R, T, B_0, B_1$.

\subsection{Evaluation of resource and performance counters} %
\label{sec:evaluation_of_counters}

Let $G_{C}({\cal S})$ be the control flow graph (CFG) of ${\cal S}$.
Hence, the statements in the basic blocks of $G_{C}({\cal S})$
are C statements, and we call such a CFG the
{\em source CFG}.
We also map ${\cal S}$ to an intermediate representation, which, itself,
is encoded in the form of a CFG, 
denoted by $G_{L}({\cal S})$, and we call it the {\em IR CFG}.
Here, we refer to the landmark textbook~\cite{Lex}
for the notion of the control flow graph
and that of intermediate representation.

We observe that ${\cal S}$ can trivially be reconstructed from 
$G_{C}({\cal S})$; hence, the knowledge of ${\cal S}$ and that of
$G_{C}({\cal S})$ can be regarded as equivalent. In contrast,
$G_{L}({\cal S})$ depends not only on ${\cal S}$ but also
on the optimization strategies that are applied to the IR
of ${\cal S}$.

Equipped with $G_{C}({\cal S})$ and $G_{L}({\cal S})$, we assume that
we can estimate each of the resource counters $r_1, \ldots, r_s$
(resp. performance counters $p_1, \ldots, p_t$) by applying functions
$f_1, \ldots, f_s$ (resp. $g_1, \ldots, g_t$) to either 
$G_{C}({\cal   S})$ or $G_{L}({\cal S})$. We call
$f_1, \ldots, f_s$ (resp. $g_1, \ldots, g_t$)
the {\em resource} (resp. {\em performance}) evaluation functions.

For instance, when ${\cal S}$ is the body of a {\cuda} kernel and
${\cal S}$ reads (resp. writes) a given array, computing the total
amount of elements read (resp. written) by one thread-block can be
determined from $G_{C}({\cal S})$.  Meanwhile, computing the minimum
number of registers to be allocated to a thread executing ${\cal S}$
requires the knowledge of $G_{L}({\cal S})$.

\subsection{Optimization strategies}   %
\label{sec:optimization_strategies}

In order to reduce the consumption of hardware resources and increase
performance counters, we assume that we have optimization procedures
$O_1, \ldots, O_w$, each of them mapping either a source CFG to
another source CFG, or an IR CFG to another IR CFG.  We
assume that the code transformations performed by $O_1, \ldots, O_w$
preserve semantics.

We associate each resource counter $r_i$, for $i = 1 \cdots s$, with a
non-empty subset ${\sigma}(r_i)$ of $\{ O_1, \ldots, O_w \}$,
such that we have 
$f_i(O({\cal S})) \leq f_i({\cal S})$, for all $O \in {\sigma}(r_i)$.
Hence, ${\sigma}(r_i)$ is a subset of the
optimization strategies among $O_1, \ldots, O_w$ that have the potential
to reduce $r_i$.
Of course, the intention is that for at least one 
$O \in {\sigma}(r_i)$, we have $f_i(O({\cal S})) < f_i({\cal S})$.
A reason for not finding such $O$ would be that ${\cal S}$
cannot be further optimized w.r.t. $r_i$.
We also make a natural {\em idempotence}  assumption:
$f_i(O(O({\cal S}))) = f_i(O({\cal S}))$,
for all $O \in {\sigma}(r_i)$.
Similarly, we associate each performance counter $p_i$, for $i = 1
\cdots t$, with a non-empty subset ${\sigma}(p_i)$  of $\{ O_1, \ldots, O_w \}$,
such that we have
$g_i(O({\cal S})) \geq g_i({\cal S})$
and 
$g_i(O(O({\cal S}))) = g_i(O({\cal S}))$, for all $O \in {\sigma}(p_i)$.
Hence, ${\sigma}(p_i)$ is a subset of the
optimization strategies among $O_1, \ldots, O_w$  that have the potential
to increase $p_i$.
The intention is, again, that
for at least one $O \in {\sigma}(p_i)$, we have
$g_i(O({\cal S})) > g_i({\cal S})$.

\subsection{Comprehensive optimization}  %
\label{sec:comprehensiveOptimization}

Let $C_1, \ldots, C_e$ be semi-algebraic systems
(that is, conjunctions of polynomial equations and inequalities)
with ${P}_1, \ldots, {P}_t$, $R_1, \ldots, R_s$, 
$D_1, \ldots, D_u$, $E_1, \ldots, E_v$ as indeterminate variables.
Let ${\cal S}_1, \ldots, {\cal S}_e$ be fragments of C programs
such that the parameters of each of them are among
$D_1, \ldots, D_u$, $E_1, \ldots, E_v$.

\begin{definition}
\label{defi:comprehensive_optimization}
The sequence $(C_1, {\cal S}_1), \ldots, (C_e,
{\cal S}_e)$ is a {\em comprehensive optimization} of ${\cal S}$ 
w.r.t. 
\begin{enumerateshort}
\item the resource evaluation functions $f_1, \ldots, f_s$,
\item the performance evaluation functions $g_1, \ldots, g_t$ and 
\item the optimization strategies  $O_1, \ldots, O_w$
\end{enumerateshort}
if the following conditions hold:
\begin{enumerateshort}
\item[$(i)$] {\rm [{\sf constraint soundness}]} 
      Each system 
      $C_1, \ldots, C_e$ is
      consistent, that is, admits at least one real solution.
\item[$(ii)$] {\rm [{\sf code soundness}]} 
    For all real values $h_1, \ldots, h_t$, $x_1, \ldots, x_s$, 
    $y_1, \ldots, y_u$, $z_1, \ldots, z_v$ of ${P}_1, \ldots, {P}_t$, 
  $R_1, \ldots, R_s$, $D_1, \ldots, D_u$, $E_1,
  \ldots, E_v$ respectively,
   for all $i \in \{1, \ldots, e\}$
   such that 
   $(h_1, \ldots, h_t$, $x_1\ldots, x_s$, $y_1, \ldots,
   y_u$, $z_1, \ldots, z_v)$ is a solution of $C_i$,
  then the code fragment
  ${\cal S}_i$ produces the same output 
   as ${\cal S}$ on any data that makes ${\cal S}$ execute correctly.
\item[$(iii)$] {\rm [{\sf coverage}]} For all real values
$y_1, \ldots, y_u, z_1, \ldots, z_v$ of 
$D_1, \ldots, D_u$, $E_1, \ldots, E_v$, respectively,
there exist $i \in \{1, \ldots, e\}$  
and real values $h_1, \ldots, h_t$, $x_1, \ldots, x_s$ of 
${P}_1, \ldots, {P}_t$, $R_1, \ldots, R_s$,
such that $(h_1, \ldots, h_t$, $x_1, \ldots, x_s$, 
$y_1, \ldots, y_u, z_1, \ldots, z_v)$
is a solution of $C_i$ and ${\cal S}_i$ produces the same output 
   as ${\cal S}$ on any data that makes ${\cal S}$ execute correctly.
\item[$(iv)$] {\rm [{\sf optimality}]} For every $i \in \{ 1, \ldots, s\}$  (resp. $\{1, \ldots, t\}$),
there exists ${\ell} \in \{1, \ldots, e\}$ such that 
for all $O \in {\sigma}(r_i)$  (resp. ${\sigma}(p_i)$) 
we have $f_i(O({\cal S}_{\ell})) = f_i({\cal S}_{\ell})$
(resp. $g_i(O({\cal S}_{\ell})) = g_i({\cal S}_{\ell})$).
\end{enumerateshort}
\label{def:Comprehensiveoptimization}
\end{definition}

We summarize Definition~\ref{def:Comprehensiveoptimization}
in non technical terms.
Condition $(i)$ states that each system of constraints
         is meaningful.
Condition $(ii)$ states
that as long as the machine, program and data parameters 
satisfy $C_i$, the code fragment ${\cal S}_i$ 
produces the same output as ${\cal S}$ on whichever
data that makes ${\cal S}$ execute correctly.
Condition $(iii)$ states that as long as ${\cal S}$ 
executes correctly on a given set of 
parameters and data, there exists  a  
code fragment ${\cal S}_i$, for suitable values of the machine parameters,
such that ${\cal S}_i$ produces the same output as ${\cal S}$
 on that set of 
parameters and data.
Finally, Condition $(iv)$ states that for each resource counter $r_i$ 
(performance counter $p_i$),
there exists at least one code fragment ${\cal S}_{\ell}$ 
for which this counter is optimal 
in the sense that it cannot be further optimized by the optimization
strategies from ${\sigma}(r_i)$  (resp. ${\sigma}(p_i)$).

\subsection{Data-structures}  %
\label{sec:data-structures}

The algorithm presented in Section~\ref{sect:mainalgo}  
computes a comprehensive optimization of ${\cal S}$ 
w.r.t. the evaluation functions $f_1, \ldots, f_s$, $g_1, \ldots, g_t$ and
optimization strategies  $O_1, \ldots, O_w$.
Hereafter, we define the main data-structure used during the course of the 
algorithm. We associate ${\cal S}$ with 
what we call a {\em quintuple}, denoted by $Q({\cal S})$
and defined as follows:
$Q({\cal S}) = ( G_{C}({\cal S}), 
                 {\lambda}({\cal S}),
                {\omega}({\cal S}),
                 {\gamma}({\cal S}),
                 C({\cal S}) )$, 
where
\begin{enumerateshort}
\item ${\lambda}({\cal S})$ is the sequence of 
      the optimization procedures among $O_1, \ldots, O_w$
      that have already been applied to the IR of ${\cal S}$; 
      hence, $G_{C}({\cal S})$ together with ${\lambda}({\cal S})$
      defines $G_{L}({\cal S})$; initially, ${\lambda}({\cal S})$ is 
      empty,
\item ${\omega}({\cal S})$ is the sequence of the optimization
      procedures among $O_1, \ldots, O_w$ that have not
      been applied so far to either $G_{C}({\cal S})$
      or $G_{L}({\cal S})$; initially, ${\omega}({\cal S})$
      is $O_1, \ldots, O_w$,
\item ${\gamma}({\cal S})$ is the sequence of resource
      and performance
      counters that remain to be evaluated on ${\cal S}$;
      initially, ${\gamma}({\cal S})$ is
      $r_1, \ldots, r_s, p_1, \ldots, p_t$,
\item $C({\cal S})$ is the sequence of the polynomial constraints 
      on 
      ${P}_1, \ldots, {P}_t$, $R_1, \ldots, R_s$, 
      $D_1, \ldots, D_u$, $E_1, \ldots, E_v$
      that have been computed so far;
      initially, $C({\cal S})$ is 
      $1 \geq {P}_1 \geq 0, \ldots, 1 \geq {P}_t \geq 0$,
      $R_1 \geq 0, \ldots, R_s \geq 0$, 
      $D_1 \geq 0, \ldots, D_u \geq 0$, $E_1 \geq 0, \ldots, E_v \geq 0$.
\end{enumerateshort}
We say that the quintuple $Q({\cal S})$ is {\em processed}
whenever ${\gamma}({\cal S})$ is empty; otherwise, 
we say that $Q({\cal S})$ is {\em in-process}.

\begin{remark}
\label{rem:stacks}
{\rm 
For the above $Q({\cal S})$, each of the sequences
${\lambda}({\cal S})$, ${\omega}({\cal S})$, ${\gamma}({\cal S})$
and $C({\cal S})$ is implemented as a stack in 
Algorithms~\ref{algo:algorithm} and \ref{algo:optimize}.
Hence, we need to specify how operations on a sequence
is performed on the corresponding stack.
Let $s_1, s_2, \ldots, s_N$ is a sequence.
\begin{enumerateshort}
\item Popping one element out of this sequence
         returns $s_1$ 
and leaves that sequence with $s_2, \ldots, s_N$,
\item Pushing an element $t_1$ on $s_1, s_2, \ldots, s_N$ 
will update that sequence to $t_1, s_1, s_2, \ldots, s_N$.
\item Pushing a sequence of elements $t_1, t_2, \ldots, t_M$ on $s_1, s_2, \ldots, s_N$ 
will update that sequence to $t_M, \ldots, t_2,  t_1, s_1, s_2, \ldots, s_N$.
\end{enumerateshort}
}
\end{remark}

\subsection{The algorithm}     %
\label{sect:mainalgo}          %
Algorithm~\ref{algo:algorithm} is the top-level procedure.
If its input is a processed quintuple $Q({\cal S})$, then it returns the pair
$(G_{C}({\cal S}), {\lambda}({\cal S}))$ (such that, after optimizing
${\cal S}$ with the optimization strategies in ${\lambda}({\cal S})$,
one can generate the IR of the optimized ${\cal S}$) together
with the system of constraints $C({\cal S})$.
Otherwise, Algorithm~\ref{algo:algorithm} is called recursively 
on each quintuple returned by ${\sf Optimize}(Q({\cal S}))$.
The pseudo-code of the {\sf Optimize} routine is given by
Algorithm~\ref{algo:optimize}.

We make a few observations about Algorithm~\ref{algo:optimize}.
\begin{enumerateshort}
\item[(R1)] Observe that at Line (5), a deep copy of the input $Q({\cal S'})$ is made,
            and this copy is called $Q({\cal S''})$.
            This duplication allows the computations to {\em fork}.
            Note that at Line (6), $Q({\cal S'})$ is modified.
\item[(R2)] In this forking process, we call $Q({\cal S'})$ the {\em accept} branch
            and $Q({\cal S''})$ the {\em refuse} branch.
            In the former case, the relation $0 \leq v_i \leq R_i$ holds thus implying
            that enough $R_i$-resources are available for executing the code fragment
            ${\cal S'}$.
            In the latter case, the relation $R_i < v_i$ holds thus implying
            that {\em not} enough $R_i$-resources are available for executing the code fragment
            ${\cal S''}$.
\item[(R3)] At Lines (18-20), a similar forking process occurs.
            Here again, we call $Q({\cal S'})$ the {\em accept} branch
            and $Q({\cal S''})$ the {\em refuse} branch.
            In the former case, the relation $0 \leq v_i \leq {P}_i$ 
            implies 
            that the ${P}_i$-performance counter may have reached its maximum 
            ratio;
            hence, no optimization strategies are applied to improve this counter.
            In the latter case, the relation ${P}_i < v_i \leq 1$ holds thus implying
            that the ${P}_i$-performance counter has not reached its maximum value;
            hence, optimization strategies are applied to improve this counter 
            if such optimization strategies are available.
             Observe that if this optimization strategy does make
             the estimated value of ${P}_i$ larger then an algebraic
             contradiction would happen and the branch will be discarded.
\item[(R4)] Line (30) in Algorithm~\ref{algo:optimize} requires non-trivial 
            computations with polynomial equations and inequalities.
            The algorithms can be found in~\cite{DBLP:journals/jsc/ChenDMMXX13}
            and are implemented in the {\RegularChains} library
            of {\Maple}.
\item[(R5)] Each system of algebraic constraints $C$ is updated
            by adding a polynomial inequality to it
            at either Lines (6), (7), (19) or (20).
            This incremental process can be performed by the
            {\tt RealTriangularize} algorithm~\cite{DBLP:journals/jsc/ChenDMMXX13}
            and implemented in the {\tt RegularChains} library.
\item[(R6)] Because of the recursive calls at Lines (16) and (29)
            several inequalities involving the same variable
            among $R_1, \ldots, R_s$, ${P}_1, \ldots, {P}_t$
            may be added to a given system $C$.
            As a result, $C$ may become inconsistent.
            For instance if $R_1 < 10$ and $10 \leq R_1$ are
            added to the same system $C$.
            This will happen when an optimization strategy
            fails to improve the value of a resource or performance
            counter.
            Note that inconstancy is automatically detected 
            by the {\tt RealTriangularize} algorithm.
\end{enumerateshort}

\begin{figure}[!htb]
\centering
\includegraphics[clip, trim=0.5cm 8cm 0.5cm 5cm, width=\linewidth]{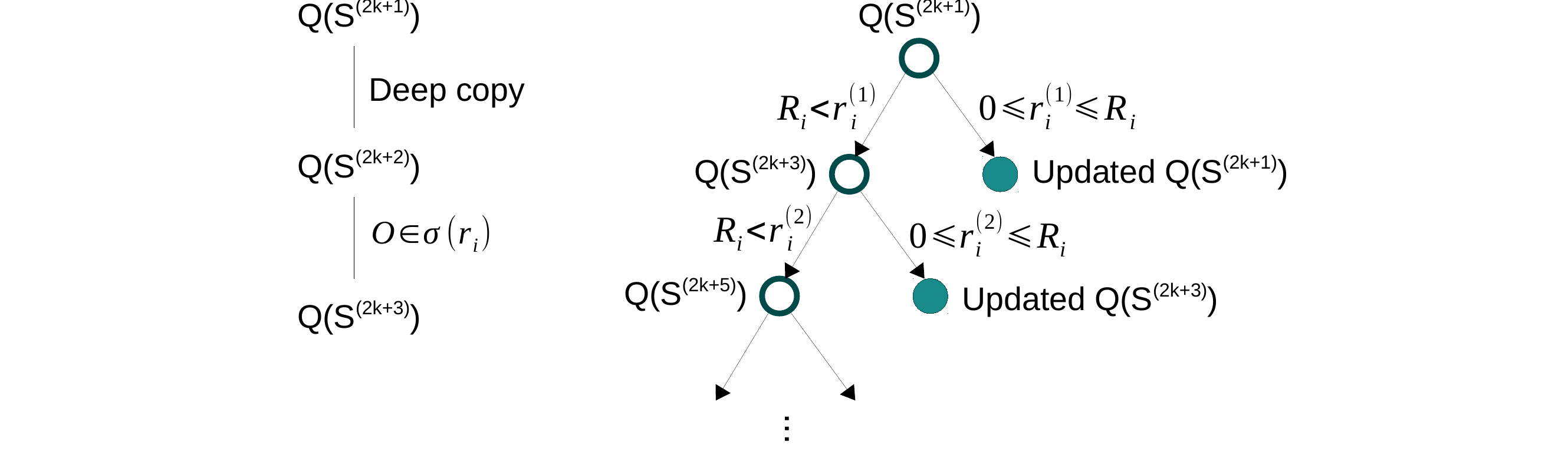}

(a) The decision subtree for resource counter $r_i$

\bigskip

\includegraphics[clip, trim=0.5cm 7cm 0.5cm 5cm, width=\linewidth]{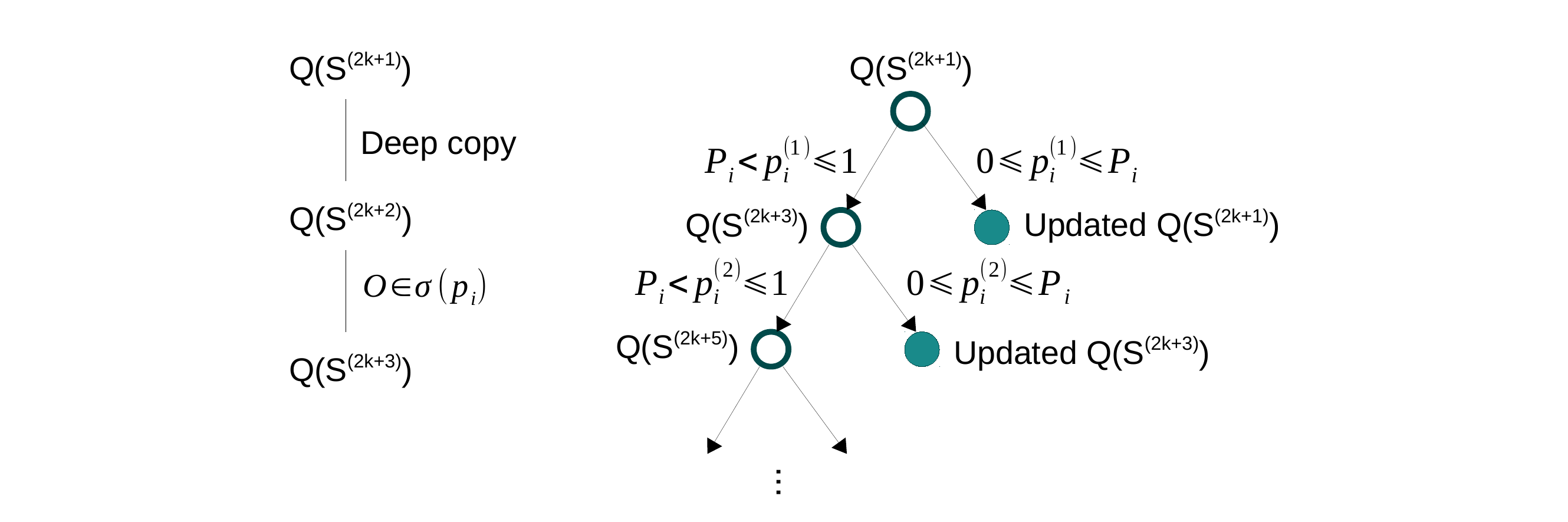}

(b) The decision subtree for performance counter $p_i$

\caption{The decision subtree for resource or performance counters}
\label{fig:countertrees}
\end{figure}

We associate the execution of Algorithm~\ref{algo:algorithm}, applied to $Q({\cal S})$, 
with a tree denoted by ${\cal T}(Q({\cal S}))$ and where
both nodes and edges of ${\cal T}(Q({\cal S}))$ are labelled.
We use the same
notations as in Algorithm~\ref{algo:optimize}.
We define ${\cal T}(Q({\cal S}))$ recursively:
\begin{enumerateshort}
\item[(T1)] We label the root of ${\cal T}(Q({\cal S}))$ with $Q({\cal S})$.
\item[(T2)] If ${\gamma}({\cal S})$ is empty, then
      ${\cal T}(Q({\cal S}))$  has no children;
      otherwise, two cases arise:
\begin{enumerateshort}
\item[(T2.1)] If no optimization strategy is to be applied for optimizing the counter $c$,
       then  ${\cal T}(Q({\cal S}))$  has a single subtree,
       which is that associated with ${\sf Optimize}(Q({\cal S'}))$
       where $Q({\cal S'})$ is obtained from $Q({\cal S})$
       by augmenting $C({\cal S})$ either with  $0 \leq v_i \leq R_i$
      if $c$ is a resource counter or with $0 \leq v_i \leq {P}_i$ otherwise.
\item[(T2.2)] If an optimization strategy is applied, then ${\cal T}(Q({\cal S}))$
      has two subtrees:
\begin{enumerateshort}
\item[(T2.2.1)] The first one is the tree associated with ${\sf Optimize}(Q({\cal S'}))$
                (where $Q({\cal S'})$ is defined as above) and is connected
                to its parent node
                by the {\em accept} edge, labelled with 
                either $0 \leq v_i \leq R_i$ or $0 \leq v_i \leq {P}_i$; see Figure~\ref{fig:countertrees}.
\item[(T2.2.2)] The second one is the tree associated with ${\sf Optimize}(Q({\cal S'''}))$
                (where $Q({\cal S'''})$ is obtained by applying the optimization strategy
                 to the deep copy of the input quintuple $Q({\cal S})$)
                 and is connected to its parent node
                by the {\em refuse} edge, labelled with 
                either $R_i < v_i$ or ${P}_i < v_i \leq 1$; see Figure~\ref{fig:countertrees}.
\end{enumerateshort}
\end{enumerateshort}
\end{enumerateshort}
Observe that every node of ${\cal T}(Q({\cal S}))$
is labelled with a quintuple and every edge 
with a polynomial constraint.

Figure~\ref{fig:countertrees} 
illustrates how Algorithm~\ref{algo:optimize}, applied to $Q({\cal S'})$, 
generates the associated tree ${\cal T}(Q({\cal S'}))$.
The cases for a {\em resource counter}  and a {\em performance counter} 
are distinguished in the sub-figures (a) and (b), respectively.
Observe that, in both cases, the accept edges {\em go south-east}, while
the refuse edges {\em go south-west}.

\begin{algorithm}[!htb]
\small
\caption{{\sf ComprehensiveOptimization}}
\label{algo:algorithm}
\Indm
\KwIn{The quintuple $Q({\cal S})$}
\KwOut{A {\em comprehensive optimization} of ${\cal S}$ 
w.r.t. the resource evaluation functions $f_1, \ldots, f_s$,
the performance evaluation functions $g_1, \ldots, g_t$ and 
the optimization strategies  $O_1, \ldots, O_w$}

\Indp

\If {\rm ${\gamma}({\cal S})$ is empty} {
  \Return $( (G_{C}({\cal S}), {\lambda}({\cal S})), C({\cal S}) )$\;
}

The output stack is initially empty\;
\For {\rm each $Q({\cal S'}) \in {\sf Optimize}(Q({\cal S}))$} {
    Push ${\sf ComprehensiveOptimization}(Q({\cal S'})$ on the output stack\;
}

\Return the output stack\;

\end{algorithm}

\begin{algorithm}[!htb]
\small
\caption{{\sf Optimize}}
\label{algo:optimize}
\Indm
\KwIn{A quintuple $Q({\cal S'})$}
\KwOut{A stack of quintuples}
\Indp

Initialize an empty stack, called {\tt result}\;
Take out from ${\gamma}({\cal S'})$ the next
resource or performance counter to be evaluated, say $c$\;
Evaluate $c$ on ${\cal S'}$ (using 
the appropriate functions among $f_1, \ldots, f_s, g_1, \ldots, g_t$)
thus obtaining a value $v_i$, which can be  
either a numerical value, 
a polynomial in ${\Q}[D_1, \ldots, D_u, E_1, \ldots, E_v]$ or
a rational function where its numerator and denominator are 
in ${\Q}[D_1, \ldots, D_u, E_1, \ldots, E_v,$ $R_1, \ldots, R_s]$\;
\If {\mbox{\rm $c$ is a resource counter $r_i$}} {
Make a deep copy $Q({\cal S''})$  of $Q({\cal S'})$,
since we are going to split the computation
into two branches:  $R_i < v_i$ and $0 \leq v_i \leq R_i$\;
Add the constraint $0 \leq v_i \leq R_i$ to $C({\cal S'})$
and push $Q({\cal S'})$ onto {\tt result}\;
Add the constraint $R_i < v_i$ to $C({\cal S''})$
and search ${\omega}({\cal S''})$ for 
an optimization strategy of ${\sigma}(r_i)$\;
\If {\rm no such optimization strategy exists} {
\Return {\tt result}\;
}
\Else {
Apply  such an optimization strategy
to $Q({\cal S''})$ yielding $Q({\cal S'''})$\;
Remove this optimization strategy from ${\omega}({\cal S'''})$\;
\If {\rm this optimization strategy is applied to the IR
of ${\cal S''}$} {
Add it to ${\lambda}({\cal S'''})$\;
}
Push $r_1, \ldots, r_{i-1}, r_i$  onto ${\gamma}({\cal S'''})$\;
Make a recursive call to {\sf Optimize} on $Q({\cal S'''})$
and push the returned quintuples onto {\tt result}\;
}
}
\If {\rm $c$ is a performance counter $p_i$} {
Make a deep copy $Q({\cal S''})$  of $Q({\cal S'})$,
since we are going to split the computation
into two branches: $0 \leq v_i \leq {P}_i$ 
and ${P}_i < v_i \leq 1$ \;
Add the constraint $0 \leq v_i \leq {P}_i$  to $C({\cal S'})$
and push $Q({\cal S'})$ onto {\tt result}\;
Add the constraint ${P}_i < v_i \leq 1$ to $C({\cal S''})$
and search ${\omega}({\cal S''})$ for 
an optimization strategy of ${\sigma}(p_i)$\;
\If {\rm no such optimization strategy exists} {
\Return {\tt result}\;
}
\Else {
Apply  such an optimization strategy
to $Q({\cal S''})$ yielding $Q({\cal S'''})$\;
Remove this optimization strategy from ${\omega}({\cal S'''})$\;
\If {\rm this optimization strategy is applied to the IR
of ${\cal S''}$} {
Add it to ${\lambda}({\cal S'''})$\;
}
Push $r_1, \ldots, r_s, p_i$  onto ${\gamma}({\cal S'''})$\;
Make a recursive call to {\sf Optimize} on $Q({\cal S'''})$
and push the returned quintuples onto {\tt result}\;
}
}
Remove from {\tt result} any quintuple
with an inconsistent system of constraints\;
\Return {\tt result}\;
\end{algorithm}

\begin{lemma}
\label{lem:terminates}
The height of the tree ${\cal T}(Q({\cal S}))$ is at most $w (s + t)$.
Therefore, Algorithm~\ref{algo:algorithm} terminates.
\end{lemma}

\begin{proof}
Consider a path $\Gamma$ from the root of ${\cal T}(Q({\cal S}))$ to
any node $N$ of ${\cal T}(Q({\cal S}))$.  Observe that $\Gamma$ counts
at most $w$ refuse edges. Indeed, following a refuse edge decreases by
one the number of optimization strategies to be used.  Observe also
that the length of every sequence of consecutive accept edges is at
most $s + t$. Indeed, following an accept edge decreases by one the
number of resource and performance counters to be evaluated.
Therefore, the number of edges in $\Gamma$ is at most $w \, (s + t)$.
\end{proof}

\begin{lemma}
\label{lem:choose_path}
Let $U := \{ U_1, \ldots, U_z \}$ be a subset of $\{ O_1, \ldots, 0_w \}$.
There exists a path from the root of ${\cal T}(Q({\cal S}))$
to a leaf of ${\cal T}(Q({\cal S}))$
along which the optimization strategies being applied are exactly those of 
$U$.
\end{lemma}

\begin{proof}
Let us start at the root of ${\cal T}(Q({\cal S}))$ and apply
the following procedure:
\begin{enumerateshort}
\item follow the refuse edge if it uses an optimization strategy
       from $\{ U_1, \ldots, U_z \}$,
\item follow the accept edge, otherwise.
\end{enumerateshort}
This creates a path from the root of ${\cal T}(Q({\cal S}))$
to a leaf with the desired property.
\end{proof}

\begin{definition}
\label{defi:optimal}
Let $i \in \{ 1, \ldots, s\}$  (resp. $\{1, \ldots, t\}$).
Let $N$ be a node of ${\cal T}(Q({\cal S}))$
and $Q({\cal S}_N)$ be the quintuple labelling this node.
We say that 
$r_i$ (resp. $p_i$) is optimal at $N$
w.r.t. the evaluation function $f_i$ (resp. $g_i$) and
the subset ${\sigma}(r_i)$  (resp. ${\sigma}(p_i)$) 
of the optimization strategies  $O_1, \ldots, O_w$,
whenever for all $O \in {\sigma}(r_i)$  (resp. ${\sigma}(p_i)$) 
we have $f_i(O({\cal S}_{N})) = f_i({\cal S}_{N})$
(resp. $g_i(O({\cal S}_{N})) = g_i({\cal S}_{N})$).
\end{definition}

\begin{lemma}
\label{lem:optimal}
Let $i \in \{ 1, \ldots, s\}$ (resp. $\{1, \ldots, t\}$).  There
exists at least one leaf $L$ of ${\cal T}(Q({\cal S}))$ such that
$r_i$ (resp. $p_i$) is optimal at $L$ 
w.r.t. the evaluation function $f_i$ (resp. $g_i$) and
the subset ${\sigma}(r_i)$  (resp. ${\sigma}(p_i)$) 
of the optimization strategies
$O_1, \ldots, O_w$.
\end{lemma}

\begin{proof}
Apply Lemma~\ref{lem:choose_path} with $U = {\sigma}(r_i)$
(resp. $U = {\sigma}(p_i)$).
\end{proof}

\begin{lemma}
Algorithm~\ref{algo:algorithm} satisfies its output specifications.
\end{lemma}

\begin{proof}
From Lemma~\ref{lem:terminates}, we know that Algorithm~\ref{algo:algorithm} 
terminates. So let $(C_1, {\cal S}_1), \ldots, (C_e,
{\cal S}_e)$ be its output.
We shall prove $(C_1, {\cal S}_1), \ldots, (C_e,
{\cal S}_e)$ satisfies the conditions $(i)$ to $(iv)$
of Definition~\ref{def:Comprehensiveoptimization}.

Condition $(i)$ is satisfied
 by the properties of the {\tt RealTriangularize}
algorithm.
Condition $(ii)$ follows clearly from the assumption
that the code transformations performed by $O_1, \ldots, O_w$
preserve semantics.
Observe that each time a polynomial inequality
is added to a system of constraints, the negation
of this inequality is also to the same system 
in another branch of the computations.
By using a simple induction on $s + t$, 
we deduce that Condition $(iii)$  is satisfied.
Finally, we prove Condition $(iv)$ by using Lemma~\ref{lem:optimal}.
\end{proof}

\section{Comprehensive Translation}
\label{comprehensive_translation}

Given a high-level model for accelerator programming (like
{\opencl}~\cite{stone2010opencl}, {\openmp}, {\openacc} or
{\meta}~\cite{DBLP:conf/iwomp/ChenMSU14}), we consider the problem of
translating a program written for such a high-level model into a
programming model for GPGPU devices, such as {\cuda}.
We assume that the numerical values of some, or all, of the hardware
characteristics of the targeted GPGPU device are unknown. 
Hence, these quantities are treated as symbols.  
Similarly, we would like that some, or all, of the program parameters remain
symbols in the generated code.

In our implementation, we focus on one high-level model for
accelerator programming, namely {\meta}. However, we believe that an
adaptation to another high-level model for accelerator programming
would not be difficult.  One supporting reason for that claim is the
fact that automatic code translation between the {\meta} and {\openmp}
languages can already be done within the {\meta} compilation
framework, see~\cite{DBLP:conf/iwomp/ChenMSU14}.

We consider as input a {\tt meta\_schedule} statement
${\cal M}$ and its surrounding {\meta} program ${\cal P}$. 
In our implementation, we assume that, apart from the 
{\tt meta\_schedule} statement ${\cal M}$, the rest of
the program ${\cal P}$ is serial C code.
Now, applying the comprehensive optimization algorithm (described in
Section~\ref{sec:comprehensive_optimization}) on the {\tt
  meta\_schedule} statement ${\cal M}$ (with prescribed resource
evaluation functions, performance evaluation functions and
optimization strategies) we obtain a sequence of processed quintuples
of {\tt meta\_schedule} statements ${\cal Q}_1({\cal M}), {\cal
  Q}_2({\cal M}), \ldots, {\cal Q}_{\ell}({\cal M})$, which forms a
comprehensive optimization in the sense of
Definition~\ref{defi:comprehensive_optimization}.

If, as mentioned in the introduction,  {\ptx} is used
as intermediate representation (IR) then, for each 
$i = 1, \ldots, {\ell}$, under the constraints
defined by the polynomial system associated with ${\cal Q}_i({\cal M})$,
the IR code associated with ${\cal Q}_i({\cal M})$ is
the translation in assembly language of a {\cuda}
counterpart of ${\cal M}$.
In our implementation, we also translate to {\cuda} source code
the {\meta} code in each ${\cal Q}_i({\cal M})$, since
this is easier to read for a human being.

\section{Experimentation}
\label{sec:experimentation}

We report on a preliminary implementation
\ifarticle
of Algorithm~\ref{algo:algorithm} 
\fi
dedicated to the optimization of {\tt
  meta\_schedule} statements in view of generating parametric {\cuda}
kernels.
Two hardware resource counters are considered: register usage per
thread, and local/shared memory allocated per thread-block. 
No performance counters are specified,
however, by design, the algorithm tries to minimize the usage
of hardware resources.
Four optimization strategies are used:
$(i)$ reducing register pressure, 
$(ii)$ controlling thread granularity,
$(iii)$ common sub-expression elimination (CSE),
and $(iv)$ caching\footnote{In the {\meta} language,
the keyword {\tt cache} is used to indicate that 
every thread accessing a specified array {\tt a}
must copy in local/shared memory the data 
it accesses in {\tt a}.} data in local/shared memory.
The first one applies to the IR CFG and uses {\llvm};
$(ii)$ and $(iii)$ are performed in {\maple} on the source CFG
while $(iv)$ combines PET~\cite{verdoolaege2012polyhedral} and 
the {\RegularChains} library in {\maple}.
Moreover, for $(i)$ and $(iii)$, 3 and 2 levels
of optimization are used, respectively.
Figure~\ref{fig:implement_tools} gives an overview of the software tools
that are involved in our implementation.

\begin{figure}[!htb]
\centering
\includegraphics[clip, trim=0.5cm 4cm 0.5cm 2cm, width=\linewidth]{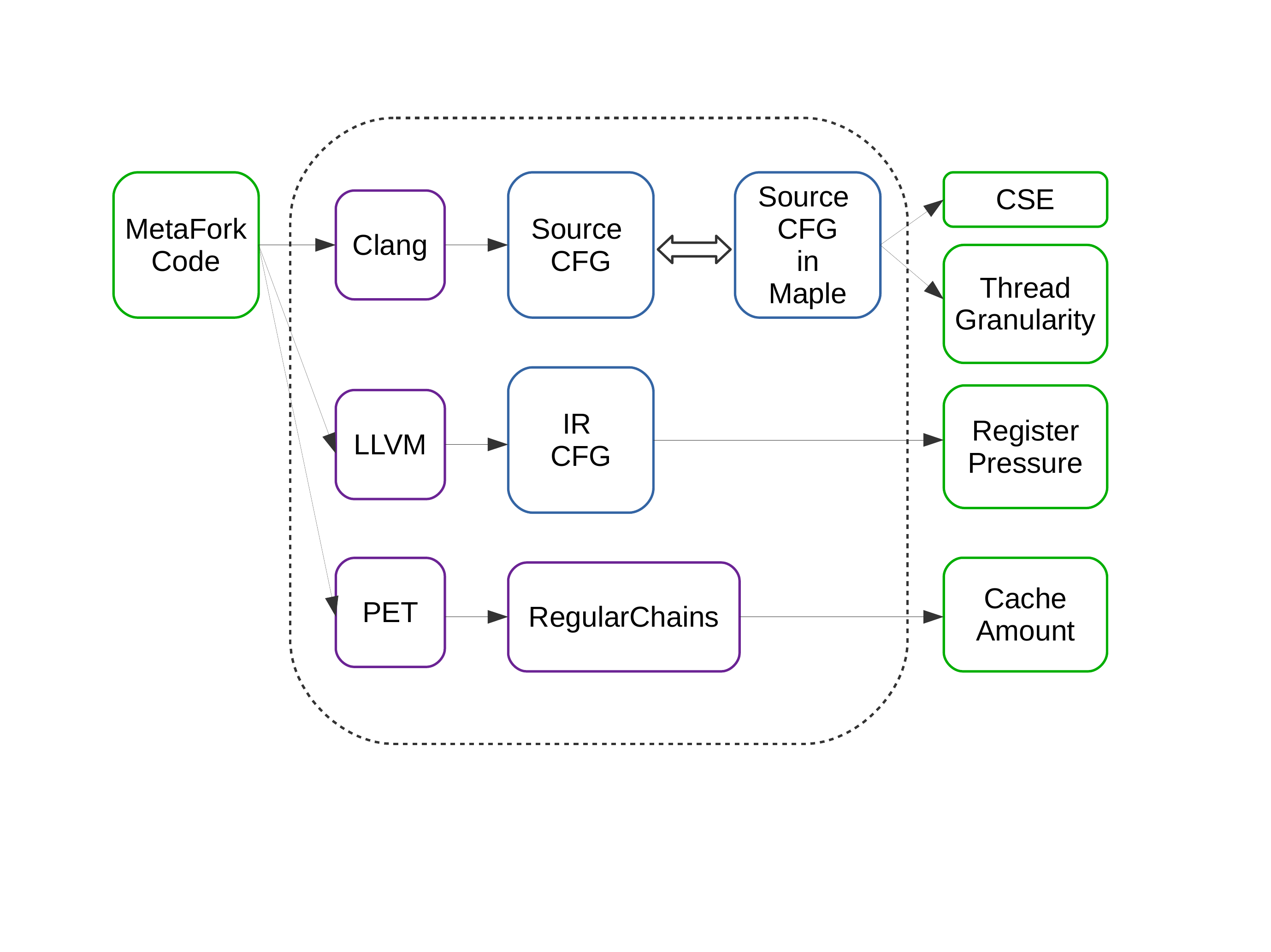}
\caption{The software tools involved in our implementation}
\label{fig:implement_tools}
\end{figure}

The test examples of our CASCON paper~\cite{DBLP:conf/cascon/ChenCKMX15} 
have been extensively tested with that implementation.
In the interest of space, we have selected
two representative examples.
For both of them, again in the interest of space, 
we present the optimized {\meta} code,
instead of the optimized  IR code
(or the {\cuda} code generated from {\meta}).

For each system of constraints, we indicate which
decisions were made by the algorithm to reach that case.
To this end, we use the following abbreviations:
(1), (2), (3a), (3b), (4a), (4b)
respectively stand for
``No register pressure optimization'',
``CSE is Applied''
``thread granularity not-reduced'',
``reduced thread granularity'',
``Use local/shared memory'',
``Do not use local/shared memory''.

For both test-examples, we give speedup factors (against 
an optiimzied serial C code) obtained with the most efficient 
of our generated {\cuda} kernels.
All {\cuda} experimental results are collected on an NVIDIA Tesla M2050.

\ifarticle
\subsection{1D Jacobi}
\fi

Both source and optimized {\meta} programs are shown on
Figure~\ref{fig:1dj}.
Table~\ref{tab:1DJacobi_granu} shows the speedup factors for the first case
of optimized {\meta} programs in Figure~\ref{fig:1dj}.

\begin{figure}[!htb]
\begin{minipage}{.45\columnwidth}
\scriptsize
\begin{center}
Source code
\end{center}

\begin{Verbatim}[commandchars=\\\{\}]
int T, N, s, B,
int dim = (N-2)/(s*B);
int a[2*N];
for (int t = 0; t<T; ++t)
 \textcolor{red}{meta_schedule} \{
  \textcolor{red}{meta_for} (int i = 0; 
              i<dim; i++)
   \textcolor{red}{meta_for} (int j = 0; 
                  j<B; j++)
    for (int k = 0; k<s; 
                     ++k) \{
     int p = i*s*B+k*B+j;
     int p1 = p + 1;
     int p2 = p + 2;
     int np = N + p;
     int np1 = N + p + 1;
     int np2 = N + p + 2;
     if (t % 2)
      a[p1] = (a[np]+
        a[np1]+a[np2])/3;
     else
      a[np1] = (a[p]+
        a[p1]+a[p2])/3;
    \}
 \}
\end{Verbatim}
\end{minipage}
\begin{minipage}{.45\columnwidth}
\scriptsize
\begin{center}
First case
\end{center}

$\left\{ \begin{array}{l}
2  s  B + 2 \leq Z_B  \\
9 \leq R_B\\
\end{array}\right.$

 (1) (4a) (3a) (2) (2)
 
\begin{Verbatim}[commandchars=\\\{\}]
for (int t = 0; t<T; ++t)
 \textcolor{red}{meta_schedule} \textcolor{violet}{cache}(a) \{
  \textcolor{red}{meta_for} (int i = 0; 
              i< dim; i++)
   \textcolor{red}{meta_for} (int j = 0; 
                  j<B; j++)
    \textcolor{blue}{for (int k = 0; k<s;} 
                     \textcolor{blue}{++k)} \{
     int p = j+(i*s+k)*B;
     int t16 = p+1;
     int t15 = p+2;
     int p1 = t16;
     int p2 = t15;
     int np = N+p;
     int np1 = N+t16;
     int np2 = N+t15;
     if (t % 2)
      a[p1] = (a[np]+
        a[np1]+a[np2])/3;
     else
      a[np1] = (a[p]+
        a[p1]+a[p2])/3;
     \}
 \}
\end{Verbatim}
\end{minipage}

\begin{minipage}{.45\columnwidth}
\scriptsize
\begin{center}
Second case
\end{center}

$\left\{ \begin{array}{l}
2  B + 2 \leq Z_B \\
Z_B <  2  s  B + 2  \\
9 \leq R_B \\
\end{array}\right.$

 (1) (3b) (4a) (3a) (2) (2)
\begin{Verbatim}[commandchars=\\\{\}]
for (int t = 0; t<T; ++t)
 \textcolor{red}{meta_schedule} \textcolor{violet}{cache}(a) \{
  \textcolor{red}{meta_for} (int i = 0; 
              i<dim; i++)
   \textcolor{red}{meta_for} (int j = 0; 
                j<B; j++) \{
    int p = i*B+j;
    int t20 = p+1;
    int t19 = p+2;
    int p1 = t20;
    int p2 = t19;
    int np = N+p;
    int np2 = N+t19;
    int np1 = N+t20;
    if (t % 2)
     a[p1] = (a[np]+
       a[np1]+a[np2])/3;
    else
     a[np1] = (a[p]+
       a[p1]+a[p2])/3;
    \}
 \}
\end{Verbatim}
\end{minipage}
\begin{minipage}{.45\columnwidth}
\scriptsize
\begin{center}
Third case
\end{center}

$\left\{ \begin{array}{l}
Z_B < 2  B + 2  \\
9 \leq R_B  \\
\end{array}\right.$

 (1) (3b) (2) (2) (4b)
\begin{Verbatim}[commandchars=\\\{\}]
for (int t = 0; t<T; ++t)
 \textcolor{red}{meta_schedule} \{
  \textcolor{red}{meta_for} (int i = 0; 
              i<dim; i++)
   \textcolor{red}{meta_for} (int j = 0; 
                j<B; j++) \{
    int p = j+i*B;
    int t16 = p+1;
    int t15 = p+2;
    int p1 = t16;
    int p2 = t15;
    int np = N+p;
    int np1 = N+t16;
    int np2 = N+t15;
    if (t % 2)
     a[p1] = (a[np]+
       a[np1]+a[np2])/3;
    else
     a[np1] = (a[p]+
       a[p1]+a[p2])/3;
    \}
 \}
\end{Verbatim}
\end{minipage}
\caption{\rm Three optimized {\meta} programs for 1D Jacobi}
\label{fig:1dj}
\end{figure}

\begin{table}[!htb]
\small
\centering
\begin{tabular}{|c|c|c|c|}
\hline 
Thread-block $\backslash$ Granularity & 2	 & 4 & 8 \\
\hline
16	& 3.340	& 4.357 & 4.975 	\\
32	& 4.785	& 5.252 & 5.206	\\
64	& 5.927	& 6.264 & 6.412	\\
128	& \textbf{10.400}	& 8.952 & 5.793	\\
256	& 6.859 	& 6.246 &	\\
\hline
\end{tabular}
\caption{\rm Speedup factors of 1D Jacobi for time iteration 4 and input vector of length  $2^{15}$+2}
\label{tab:1DJacobi_granu}
\end{table}

\ifarticle
\subsection{Matrix transposition}
\fi

Three cases of optimized {\meta} programs are 
shown on Figure~\ref{fig:transpose}.
Table~\ref{tab:transpose_granu} shows the speedup factors for
the first case of optimized {\meta} programs in Figure~\ref{fig:transpose}.

\begin{figure}[!htb]
\begin{minipage}{.3\columnwidth}
\begin{scriptsize}
\begin{center}
First case
\end{center}

$\left\{ \begin{array}{l}
2  s  B_0 B_1 \leq Z_B  \\
6 \leq R_B\\
\end{array}\right.$

 (1) (4a) (3a) (2) (2) 
\end{scriptsize}
\end{minipage}
\begin{minipage}{0.5\columnwidth}
\begin{scriptsize}
\begin{Verbatim}[commandchars=\\\{\}]
\textcolor{red}{meta_schedule} \textcolor{violet}{cache}(a, c) \{
 \textcolor{red}{meta_for} (int v0 = 0; v0<dim0; v0++)
  \textcolor{red}{meta_for} (int v1 = 0; v1<dim1; v1++)
   \textcolor{red}{meta_for} (int u0 = 0; u0<B0; u0++)
    \textcolor{red}{meta_for} (int u1 = 0; u1<B1; u1++)
     \textcolor{blue}{for (int k = 0; k < s; ++k)} \{
      int i = v0*B0+u0;
      int j = (v1*s+k)*B1+u1;
      c[i*N+j] = a[i][j];
     \}
\}
\end{Verbatim}
\end{scriptsize}
\end{minipage}

\begin{minipage}{.3\columnwidth}
\begin{scriptsize}
\begin{center}
Second case
\end{center}

$\left\{ \begin{array}{l}
2  B_0  B_1 \leq Z_B \\
Z_B <  2  s  B_0  B_1  \\
5 \leq R_B \\
\end{array}\right.$

 (1) (3b) (4a) (3a) (2) (2) 
\smallskip

$\left\{ \begin{array}{l}
2  B_0  B_1 \leq Z_B \\
Z_B  <  2  s  B_0  B_1  \\
5 \leq R_B < 6 \\
\end{array}\right.$

 (2) (2) (3b) (1) (4a) (3a) 
\end{scriptsize}
\end{minipage}
\begin{minipage}{0.5\columnwidth}
\begin{scriptsize}
\begin{Verbatim}[commandchars=\\\{\}]
\textcolor{red}{meta_schedule} \textcolor{violet}{cache}(a, c) \{
 \textcolor{red}{meta_for} (int v0 = 0; v0<dim0; v0++)
  \textcolor{red}{meta_for} (int v1 = 0; v1<dim1; v1++)
   \textcolor{red}{meta_for} (int u0 = 0; u0<B0; u0++)
    \textcolor{red}{meta_for} (int u1 = 0; u1<B1; u1++) \{
     int i = v0*B0+u0;
     int j = v1*B1+u1;
     c[i*N+j] = a[i][j];
    \}
\}
\end{Verbatim}
\end{scriptsize}
\end{minipage}

\begin{minipage}{.3\columnwidth}
\begin{scriptsize}
\begin{center}
Third case
\end{center}

$\left\{ \begin{array}{l}
Z_B < 2  B_0  B_1 \\
5 \leq R_B  \\
\end{array}\right.$

 (1) (3b) (2) (2) (4b) 
\smallskip

$\left\{ \begin{array}{l}
Z_B < 2  B_0  B_1 \\
5 \leq R_B < 6  \\
\end{array}\right.$

 (2) (2) (3b) (1) (4b) 
\end{scriptsize}
\end{minipage}
\begin{minipage}{0.5\columnwidth}
\begin{scriptsize}
\begin{Verbatim}[commandchars=\\\{\}]
\textcolor{red}{meta_schedule} \{
 \textcolor{red}{meta_for} (int v0 = 0; v0<dim0; v0++)
  \textcolor{red}{meta_for} (int v1 = 0; v1<dim1; v1++)
   \textcolor{red}{meta_for} (int u0 = 0; u0<B0; u0++)
    \textcolor{red}{meta_for} (int u1 = 0; u1<B1; u1++) \{
     int i = v0*B0+u0;
     int j = v1*B1+u1;
     c[i*N+j] = a[i][j];
    \}
\}
\end{Verbatim}
\end{scriptsize}
\end{minipage}
\caption{\rm Three optimized {\meta} programs for matrix transposition}
\label{fig:transpose}
\end{figure}


\begin{table}[!htb]
\small
\centering
\begin{tabular}{|c|c|c|c|}
\hline
Thread-block $\backslash$ Granularity & 2 & 4 & 8 \\
\hline
(4, 32) & 103.281 & 96.284  & 75.211  \\
(8, 32) & \textbf{111.971}  & 90.625  & 85.422  \\
(16, 32) &78.476  & 68.894  & 48.822  \\
(32, 32) & 45.084  & 46.425  & 32.824  \\
\hline
\end{tabular}
\caption{\rm Speedup factors of matrix transposition for input matrix of order  $2^{14}$}
\label{tab:transpose_granu}
\end{table}

\ifshow

\subsection{Polynomial division}
Three cases of optimized {\meta} programs are shown on Figure~\ref{fig:division}.
Table~\ref{tab:pd} shows the speedup factors for the third case of
optimized {\meta} programs in Figure~\ref{fig:division}, which outperforms
the first case.

\begin{figure}[!htb]
\begin{minipage}{.3\columnwidth}
\begin{scriptsize}
\begin{center}
First case
\end{center}

$\left\{ \begin{array}{l}
2 s B + 2 \leq Z_B  \\
5 \leq R_B\\
\end{array}\right.$

 (1) (4a) (3a) (2) (2) 
\end{scriptsize}
\end{minipage}
\begin{minipage}{0.5\columnwidth}
\begin{scriptsize}
\begin{Verbatim}[commandchars=\\\{\}]
while(N >= M) \{
 \textcolor{red}{meta_schedule} \textcolor{violet}{cache}(a, b) \{
  \textcolor{red}{meta_for} (int u = 0; u<dim; ++u)
   \textcolor{red}{meta_for} (int v = 0; v<B; ++v)
    \textcolor{blue}{for (int k = 0; k < s; ++k)} \{
     int p = (u*s+k)*B+v;
     int x = N-M+p-1;
     int e = (a[N-1] * b[p]) % P;
     a[x] = (a[x] - e) % P;
   \}
 \}
 N--;
\}
\end{Verbatim}
\end{scriptsize}
\end{minipage}

\begin{minipage}{.3\columnwidth}
\begin{scriptsize}
\begin{center}
Second case
\end{center}

$\left\{ \begin{array}{l}
2 B + 2 \leq Z_B  \\
Z_B <  2 s B + 2  \\
4 \leq R_B < 5 \\
\end{array}\right.$

 (1) (3b) (4a) (3a) (2) (2) 
\smallskip

$\left\{ \begin{array}{l}
2 B + 2 \leq Z_B \\
Z_B <  2 s B + 2  \\
4 \leq R_B < 5 \\
\end{array}\right.$

 (2) (2) (3b) (1) (4a) (3a) 
\end{scriptsize}
\end{minipage}
\begin{minipage}{0.5\columnwidth}
\begin{scriptsize}
\begin{Verbatim}[commandchars=\\\{\}]
while(N >= M) \{
 \textcolor{red}{meta_schedule} \textcolor{violet}{cache}(a, b) \{
  \textcolor{red}{meta_for} (int u = 0; u<dim; ++u)
   \textcolor{red}{meta_for} (int v = 0; v<B; ++v) \{
    int p = u*B+v;
    int x = N-M+p-1;
    int e = (a[N-1] * b[p]) % P;
    a[x] = (a[x] - e) % P;
   \}
 \}
 N--;
\}
\end{Verbatim}
\end{scriptsize}
\end{minipage}

\begin{minipage}{.3\columnwidth}
\begin{scriptsize}
\begin{center}
Third case
\end{center}

$\left\{ \begin{array}{l}
Z_B < 2 B + 2  \\
4 \leq R_B < 5 \\
\end{array}\right.$

 (1) (3b) (2) (2) (4b) 
\smallskip

$\left\{ \begin{array}{l}
Z_B < 2 B + 2  \\
4 \leq R_B < 5 \\
\end{array}\right.$

 (2) (2) (3b) (1) (4b) 
\end{scriptsize}
\end{minipage}
\begin{minipage}{0.5\columnwidth}
\begin{scriptsize}
\begin{Verbatim}[commandchars=\\\{\}]
while(N >= M) \{
 \textcolor{red}{meta_schedule} \{
  \textcolor{red}{meta_for} (int u = 0; u<dim; ++u)
   \textcolor{red}{meta_for} (int v = 0; v<B; ++v) \{
    int p = u*B+v;
    int x = N-M+p-1;
    int e = (a[N-1] * b[p]) % P;
    a[x] = (a[x] - e) % P;
   \}
 \}
 N--;
\}
\end{Verbatim}
\end{scriptsize}
\end{minipage}
\caption{\rm Three optimized {\meta} program for polynomial division}
\label{fig:division}
\end{figure}

\begin{table}[!htb]
\small
\centering
\begin{tabular}{|c|c|c|c|}
\hline
& \multicolumn{3}{|c|}{Input-size {\tt M}, {\tt N} = {\tt M}+100} \\
\hline
Thread-block & $2^{13}$ & $2^{14}$ & $2^{15}$ \\
\hline
16  & 5.468 & 7.094 & 8.433 \\
32  & 7.208 & 10.758 & 12.739 \\
64  & 8.375 & 14.057 & 20.974 \\
128 & \textbf{9.245} & \textbf{15.688} & \textbf{23.881} \\
256 & 8.094 & 15.316 & 23.115 \\
512 & 7.362 & 14.201 & 22.565 \\
\hline
\end{tabular}
\caption{\rm Speedup factors of polynomial division in {\meta} without the granularity} 
\label{tab:pd}
\end{table}

\fi

\section{Concluding Remarks}
\label{sec:Conclusion}

We have shown how, from an annotated C/C++ program, parametric {\cuda}
kernels could be optimized.  These optimized parametric {\cuda}
kernels are organized in the form of a case discussion, where cases
depend on the values of machine parameters (e.g. hardware resource
limits) and program parameters (e.g. dimension sizes of
thread-blocks). 

The proposed approach extend previous works,
in particular {\ppcg}~\cite{DBLP:journals/taco/VerdoolaegeJCGTC13} and
{\cudachill}~\cite{Khan:2013:SAC:2400682.2400690}, and combine them with
techniques from computer algebra. Indeed, handling systems of 
non-linear polynomial equations and inequalities is required
in the context of parametric {\cuda} kernels.

Our preliminary implementation uses {\llvm}, {\Maple}  and 
{\ppcg}; it successfully processes a variety of standard test-examples.
In particular, the computer algebra portion of the computations
is not a bottleneck.

\section*{Acknowledgments}
The authors would like to thank the IBM Toronto labs and
NSERC of Canada for supporting their work.

\begin{small}
\bibliography{references}
\bibliographystyle{plain}
\end{small}

\end{document}